\documentclass{article}
\usepackage{mathrsfs}
\usepackage{a4wide}
\usepackage{amsmath,amsfonts,amssymb,amsthm}

\def \beq{\begin{equation}}
\def \eeq{\end{equation}}

\def\and {{\rm \; and \;}}

\newcommand{\R}{{\mathbb R}}

\newcommand{\y}{{\bf y}}

\newcommand{\x}{{\bf x}}

\newcommand{\bk}{{\bf k}}

\newtheorem{theorem}{Theorem}[section]

\newtheorem{proposition}[theorem]{Proposition}
\newtheorem{corollary}[theorem]{Corollary}
\newtheorem{lemma}[theorem]{Lemma}

\theoremstyle{definition}

\numberwithin{equation}{section}

\begin{document}

\noindent 
\begin{center}
\textbf{\large On the Lipschitz continuity of spectral bands of
  Harper-like and  magnetic Schr\"odinger operators}
\end{center}

\begin{center}
June 22, 2010
\end{center}

\vspace{0.5cm}

\begin{center}
\textbf{ 
Horia D. Cornean\footnote{Department of Mathematical Sciences,
  Aalborg University, Fredrik Bajers Vej 7G, 9220 Aalborg, Denmark}}

\end{center}

\begin{abstract}
We show for a large class of discrete Harper-like and continuous
magnetic Schr\"odinger operators that their band edges are Lipschitz
continuous with respect to the intensity of the external constant
magnetic field. We generalize a result obtained by
J. Bellissard in 1994, and give examples in favor of a recent conjecture of
G. Nenciu.
\end{abstract}

\noindent

\section{Introduction and the main results}
\noindent {\it  Harper-like operators}. Let $\Gamma\subset
\mathbb{\R}^2$ be 
a (possibly irregular) lattice which has the property that there
exists an injective map 
$F:\Gamma\mapsto \mathbb{Z}^2$ such that 
$|F(\gamma)-\gamma|<1/2$.  The Hilbert space is $l^2(\Gamma)$. 

The elements of the canonical basis in $l^2(\Gamma)$ are denoted by
$\{\delta_{\x}\}_{\x\in\Gamma}$, where $\delta_{\x}(\y)=1$ if $\y=\x$ and zero
otherwise. In the discrete case, to any bounded self-adjoint operator 
$H\in B(l^2(\Gamma))$ it corresponds a bounded and symmetric kernel 
$H(\x,\x')=\langle H\delta_{\x'},\delta_{\x}\rangle=\overline{H(\x',\x)}$. We will
extensively use the Schur-Holmgren upper bound for the norm of a
self-adjoint operator:
\begin{equation}\label{novem1}
||H||\leq \sup_{\x'\in \Gamma}
\sum_{\x\in\Gamma}
  |H(\x,\x')|.
\end{equation}

Denote by 
$\langle \x
-\x_0\rangle^\alpha=[1+(\x-\x_0)^2]^{\frac{\alpha}{2}}$, $\alpha \geq 0$. We 
define $\mathcal{C}^\alpha$ to be the set of bounded
and self-adjoint operators $H\in B(l^2(\Gamma))$ which have the
property that their kernels obey a weighted Schur-Holmgren type estimate: 
\begin{equation}\label{febr1}
||H||_{\mathcal{C}^\alpha}:=\sup_{\x'\in \Gamma}
\sum_{\x\in\Gamma}\langle \x -\x'\rangle^\alpha
  |H(\x,\x')| <\infty.
\end{equation}

We also define the space $\mathcal{H}^\alpha$ which 
contains bounded and self-adjoint
operators $H$ which obey:
\begin{equation}\label{febr2}
||H||_{\mathcal{H}^\alpha}:=
\sup_{\x'\in \Gamma}\left \{\sum_{\x\in \Gamma}\langle \x -\x'\rangle^{2\alpha}
  |H(\x,\x')|^2\right \}^{\frac{1}{2}}<\infty.
\end{equation}
The flux of a unit magnetic field
orthogonal to the plane through a triangle generated by 
$\x$, $\x'$ and the origin is given by:
\begin{align}\label{iunie1}
\varphi({\bf x},{\bf x}'):=
-\frac{1}{2}\left (x_1\:x_2'-x_2\:x_1'\right )=-\varphi(\x',\x).
\end{align}
Note the important additive identity:
\begin{align}\label{iunie2}
\varphi({\bf x},\y)+\varphi(\y,\x')&=\varphi(\x,\x')+
\varphi(\x-\y,\y-\x'),\\
|\varphi(\x-\y,\y-\x')|&\leq \frac{1}{2}|\x-\y|\;|\y-\x'|.\nonumber
\end{align}

Let $K\in \mathcal{C}^0$. Let its 
kernel be $K(\x,\x')$. We are interested in a family of Harper-like operators 
$\{K_b\}_{b\in\R}$ given by the kernels $e^{ib
  \varphi(\x,\x')}K(\x,\x')$. Clearly, $\{K_b\}_{b\in\R}\subset 
\mathcal{C}^0$. The usual Harper operator lives in $l^2(\mathbb{Z}^2)$,
and its generating kernel has the form $K(\x,\x')=k(\x-\x')$ where $k(\x)$ equals $1$ if $|\x|=1$, 
and $0$ otherwise.

In Lemma \ref{lema1} we will show that
$\mathcal{H}^\alpha\subset\mathcal{C}^0 $ if $\alpha>1$. Now 
here is the first main result of our paper: 
\begin{theorem}\label{teorema1} Let $\alpha >3$ and  
$K\in \mathcal{H}^\alpha$. Construct the
corresponding family of Harper-like operators $\{K_b\}_{b\in\R}$. Then we have:

{\rm i}. The resolvent set $\rho(K_b)$
is stable; more precisely, if ${\rm dist}(z,\sigma(K_{b_0}))\geq \epsilon>0$
then there exist $\delta>0$ and $\eta>0$ such that ${\rm
  dist}(z,\sigma(K_{b}))\geq \eta$ whenever $|b-b_0|<\delta$. 

{\rm ii}.  Define $E_+(b):=\sup
\sigma(K_b)$ and $E_-(b):=\inf
\sigma(K_b)$. Then $E_\pm$ are Lipschitz functions of $b$. 

{\rm iii}. Let $\alpha>4$. Assume that $K_{b_0}$ has a gap in the
spectrum of the form $(e_-(b_0),e_+(b_0))$, where $e_\pm(b_0)\in
\sigma(K_{b_0})$ are the gap edges. Then as long as the gap is not
closing by varying $b$ in a closed interval $I$ containing $b_0$, the
operator $K_{b}$ will have a gap 
$(e_-(b),e_+(b))$ whose edges are Lipschitz functions of $b$ on $I$.
\end{theorem}

\vspace{0.5cm}

{\bf Remark}. Denoting by $\delta b=b-b_0$, then according to our notations 
we have that $K_b=\left (K_{b_0}\right )_{\delta b}$. It means that it
is enough to prove spectral stability and Lipschitz properties
near $b_0=0$.

\vspace{0.5cm}

We can complicate the setting by allowing the
generating kernel to depend on $b$. 
\begin{corollary}\label{corolar1}
Assume that the generating kernel $K(\x,\x';b)$ obeys all the
spatial localization 
conditions of Theorem \ref{teorema1}, uniformly in $b\in \R$. Moreover,
assume that it also satisfies an extra 
condition: 
\begin{equation}\label{aprilie1}
 \sup_{\x'\in\Gamma}\sum_{\x\in\Gamma}|K(\x,\x';b)-K(\x,\x';b_0)|\;
 \leq C\;|b-b_0|,\quad |b-b_0|\leq 1.
\end{equation}
Consider the family $\{K_b\}_{b\in\R}$ generated by $e^{ib\varphi(\x,\x')}
 K(\x,\x';b)$. Then Theorem \ref{teorema1} holds true for $K_b$.  
\end{corollary}
\vspace{0.5cm}

\noindent {\it Continuous Schr\"odinger operators.} 
Let us consider the operator in $L^2(\R^2)$
\begin{align}\label{hamy}
H(b):=({\bf p}-b{\bf a})^2+V,\quad {\bf p}=-i\nabla_{\bf
  x},\quad{\bf a}({\bf x})=(-x_2/2,x_1/2 ),\quad b\in\R.
\end{align}
where we assume that the scalar potential $V$ is smooth and bounded
together with all its derivatives on $\R^2$. This very strong condition is
definitely not necessary for 
the result given
below, but it simplifies the presentation. For the same reason we
formulate the result only near $b_0=0$.
\begin{theorem}\label{teorema2}
Assume that the spectrum of $H(0)$ has a finite and isolated 
spectral band $\sigma_0$, where 
$\sigma_0=[s_-(0),s_+(0)]$. Then if $|b|$ is small enough, $\sigma_0$
will evolve into a still isolated spectral island $\sigma_b\subset
\sigma(H(b))$. Denote by $s_-(b):=\inf \sigma_b$ and
$s_+(b):=\sup\sigma_b$. Then these edges are Lipschitz at $b=0$,
i.e. there exists a constant
$C$ such that $|s_\pm(b)-s_\pm(0)|\leq C\; |b|$.  
\end{theorem}

\vspace{0.5cm}

{\bf Remark}. We do not exclude the appearance of gaps
inside $\sigma_b$. Moreover, 
the formulation of this result is slightly different from the
one we gave in the discrete case. Here we look at the edges of a 
finite part of the spectrum,
and not at the edges of a gap. In the discrete case both
formulations are equivalent. However, our proof does not work in the
continuous case if $\sigma_0$ is infinite.

\subsection{Previous results and open problems}

Spectrum stability is a fundamental issue in perturbation
theory. It is well known that if $W$ is relatively bounded to
$H_0$, then the spectrum of $H_\lambda=H_0+\lambda W$ is at a Hausdorff distance
of order $|\lambda|$ from the spectrum of $H_0$. But this is in
general not true for perturbations which are not relatively
bounded. And the magnetic perturbation coming from a constant field 
is not relatively bounded, neither in the discrete nor in the 
continuous case.    

With the notable exception of a recent paper by Nenciu \cite{Nen3}, 
all previous results on the discrete case we are aware of deal with the
situation in which $\Gamma=\mathbb{Z}^2$ and the generating kernel obeys 
$K(\x,\x')=k(\x-\x')$, where $k$ is sufficiently fast decaying at infinity. 
Maybe the first proof of spectral stability of
Harper operators is due to Elliott \cite{Ell}. The result is
refined in \cite{BEY} where it is shown that the gap boundaries are
$\frac{1}{3}$-H\"older continuous in $b$. Later results by Avron, van
Mouche and Simon \cite{AMS}, Helffer and Sj\"ostrand \cite{He-Sj1,
  He-Sj2}, and Haagerup and R{\o}rdam \cite{HR} 
pushed the exponent up to $\frac{1}{2}$. In fact they prove
more, they show that the Hausdorff distance between spectra behaves like
$|b-b_0|^{\frac{1}{2}}$. These results are
optimal in the sense that the H\"older constant is independent of the
length of the eventual gaps, and it  is known that these gaps can close down
precisely like $|b-b_0|^{\frac{1}{2}}$ 
near rational values of $b_0$ \cite{He-Sj2, HKS}. Note that Nenciu
\cite{Nen3} proves a similar result for a much larger class of
Harper-like operators. Many other spectral properties of
Harper operators can be found in a paper by Herrmann and Janssen \cite{HJ}. 

In the continuous case, the stability of gaps was first shown by Avron
and Simon \cite{AS}, and Nenciu \cite{Nen2}. Nenciu's result implicitly
gives a $\frac{1}{2}$-H\"older continuity in $b$ for the Hausdorff
distance between spectra.  Then in \cite{BC} the H\"older exponent of
gap edges was pushed up to $\frac{2}{3}$.   

The first proof of Lipschitz continuity of gap edges 
for Harper-like operators 
was given by Bellissard \cite{Bell} (later on Kotani
\cite{Ko} extended his method to more general regular lattices and
dimensions larger than two). The configuration space is $\Gamma=\mathbb{Z}^2$ and the generating kernel is of the form 
$K(\x,\x')=k(\x-\x';b)$, where $k(\x;b)$ decays polynomially in  $|\x|$ and is allowed to depend smoothly on $b$. This extra-dependence is not central for our discussion, so we will consider that $k$ is $b$ independent. Bellissard's innovative idea uses in an essential way that
the Harper operators generated by translation invariant and fast
decaying kernels $k(\x-\x')$ can be written as linear
combinations of magnetic translations: 
$$K_b=\sum_{\gamma\in\mathbb{Z}^2}k(\gamma)W_b(\gamma),\quad
[W_b(\gamma)\psi](\x)=e^{ib\varphi(\x,\gamma)}\psi(\x-\gamma),\quad
W_b(\gamma)W_b(\gamma')=e^{ib\varphi(\gamma,\gamma')}W_b(\gamma+\gamma').$$   
Bellissard's crucial observation was that the $C^*$ algebra
$\mathcal{A}_{b_0+\delta}$ generated by
$\{W_{b_0+\delta}(\gamma)\}_{\gamma\in\mathbb{Z}^2} $ is isomorphic
with a sub-algebra of 
$\mathcal{A}_{b_0}\otimes \mathcal{A}_{\delta}$ which is generated by $
\{W_{b_0}(\gamma)\otimes
W_{\delta}(\gamma)\}_{\gamma\in\mathbb{Z}^2}$. Thus one can construct
an operator $\widetilde{K}_{b_0+\delta}$ which is isospectral with
$K_{b_0+\delta}$. The new operator lives in the space
$l^2(\mathbb{Z}^2)\otimes L^2(\R)$, and
$\widetilde{K}_{b_0}=K_{b_0}\otimes {\rm Id}$. It turns out that it is more convenient
to study the spectral edges of the new operator. The reason is that
the singularity induced by the magnetic perturbation is hidden in the
extra-dimension. But the proof breaks down in case of irregular lattices or if
the generating kernel $K(\x,\x')$ is not just a function of $\x-\x'$. 

Coming back to our proof, its crucial ingredient consists in expressing the
magnetic phases with the help of the heat kernel of 
{\it a continuous Schr\"odinger operator}, see
\eqref{ianu1}-\eqref{ianu4}. Moreover, the proof in the discrete case
also works for continuous kernels living on $\R^2$ and not just on
lattices.  This is what we use in the last step of the proof of Theorem
\ref{teorema2} dealing with continuous magnetic Schr\"odinger operators.

A limitation of our method consists in the fact that the
phases $\varphi(\x,\x')$ are generated by a constant magnetic field. A
 more general discrete problem was formulated by Nenciu in \cite{Nen3}
where he proposed to replace the explicit formulas in \eqref{iunie1}
and \eqref{iunie2} with more general real and antisymmetric phases obeying 
$\phi(\x,\x')=\overline{\phi(\x,\x')}=-\phi(\x',\x)$ and 
$$|\phi(\x,\y)+\phi(\y,\x')+\phi(\x',\x)|\leq {\rm
  area}\;\Delta(\x,\y,\x')$$
where $\Delta(\x,\y,\x')$ is the triangle generated by the three
points. These phases appear very naturally in the continuous case, see
\cite{CN, CN2, IMP, LMS, MP1, MP2, MP3, Nen1}, where it is shown that 
if ${\bf a}(\x)$
is the transverse gauge generated by a globally bounded 
magnetic field $|b(\x)|\leq 1$, then 
$\phi(\x,\x')$ can be chosen to be the path integral of 
${\bf a}(\x)$ on the segment linking $\x'$ with $\x$. This is the
same as the magnetic flux of $b$ through the triangle generated by $\x$, $\x'$
and the origin. 

Using a completely different proof method, Nenciu shows among other
things in \cite{Nen3} that the gap edges are Lipschitz up to a
logarithmic factor, and he conjectures that they are actually
Lipschitz. His method relies on the theory of almost convex
functions, and the result provided by this technique is optimal in the
sense that it cannot be improved in order to get rid of the
logarithm. A new idea would be necessary in order to prove Nenciu's
Lipschitz conjecture. 

Our current paper supports this conjecture because it provides examples of
phases not coming from a constant magnetic field which still generate 
Lipschitz gap edges. Let us show this here. 

Consider an irregular lattice $\Gamma\subset \R^2$ which is a 
local deformation of
$\mathbb{Z}^2$, that is there exists a bijective map
$F:\Gamma\to \mathbb{Z}^2$ such that $|F(\gamma)-\gamma|<\frac{1}{2}$. Define
the phases
$\widetilde{\varphi}(\x,\x'):=\varphi(F^{-1}(\x),F^{-1}(\x'))$ where
$\varphi$ is given by \eqref{iunie1}.

Choose any self-adjoint operator $K\in B(l^2(\mathbb{Z}^2))$ given by a
kernel $K(\x,\x')$
sufficiently fast decaying outside the diagonal. The same operator
can be seen in $B(l^2(\Gamma))$ given by
$\widetilde{K}(\gamma,\gamma'):=K(F(\gamma),F(\gamma'))$.  
Thus the operator $K_b$ generated by ${K}_b(\x,\x'):=e^{ib\widetilde{\varphi}(\x,\x')}K(\x,\x')$ is
unitary equivalent with an operator in $B(l^2(\Gamma))$ with a kernel 
$$\widetilde{K}_b(\gamma,\gamma'):=e^{ib\varphi(\gamma,\gamma')}\widetilde{K}(\gamma,\gamma').$$
In this case, we know from Theorem \ref{teorema1} 
that the edges of the spectral gaps of $\widetilde{K}_b$ and thus ${K}_b$ will
have a Lipschitz behavior. But the general case remains open.

\section{Proof of Theorem \ref{teorema1}}

This section is dedicated to the proof of our first theorem. Parts of
this proof will be later on adapted to the continuous case in Theorem
\ref{teorema2}.  

\subsection{Proof of {\it (i)}}

Let us start by showing the existence of natural embeddings of
$\mathcal{C}^\alpha$'s in $\mathcal{H}^\alpha$'s  given by the
following short lemma:
\begin{lemma}\label{lema1}
Let $H\in \mathcal{H}^\alpha$ with $\alpha>1$. Then $H\in \mathcal{C}^\beta$
with $\beta<\alpha-1$. In particular, if $\alpha>3$ then the kernel 
$\langle \x -\x'\rangle^{2}|H(\x,\x')|$ obeys a Schur-Holmgren
estimate and thus defines a bounded operator.
\end{lemma}
\begin{proof}Choose some small enough 
  $\epsilon>0$ such that $\alpha>\beta+1+\epsilon$. We write:
$$\langle \x -\x'\rangle^{\beta}|H(\x,\x')|\leq \langle \x
-\x'\rangle^{-1-\epsilon}\langle \x
-\x'\rangle^{\alpha}|H(\x,\x')|
$$
and see that the Cauchy-Schwarz inequality 
gives 
\begin{equation}\label{april1}
||H||_{ \mathcal{C}^\beta}\leq C_{\alpha,\beta}||H||_{\mathcal{H}^\alpha} \;. 
\end{equation}
\end{proof}

\vspace{0.5cm}

Another technical estimate to be proved in the
Appendix claims that if $H$ has a kernel which is localized
near the diagonal, then the resolvent's kernel will also have such a
localization. 
\begin{proposition}\label{prop1}
Let $H\in \mathcal{C}^\alpha$, with $\alpha > 0$. Let 
$z\in \rho(H)$. Then for every $0\leq \alpha'<\alpha$ we have $(H-z)^{-1}\in 
\mathcal{H}^{\alpha'}$, and there exists a constant $C$ 
independent of $z$ such that  
\begin{align}\label{april8}
||(H-z)^{-1}||_{\mathcal{H}^{\alpha'}}\leq C\;(1+||H||_{\mathcal{C}^\alpha}^{\alpha+1})\left (
\frac{1}
{\{{\rm dist}(z,\sigma(H))\}^{\alpha+2}}+\frac{1}{{\rm
    dist}(z,\sigma(H))}\right ).
\end{align}
\end{proposition}

\noindent{\bf Remark}. This proposition is related to what specialists in von
Neumann algebras would call dual action, see \cite{Tak1, Tak2,
  Take}. Stronger localization results have been earlier obtained by
Jaffard \cite{Jaffa}, later generalized by Gr\"ochenig and Leinert
\cite{GL}.  We choose for completeness to give an elementary proof in the
Appendix; our proof also highlights the uniformity in $z\in
\rho(H)$. 

\vspace{0.5cm}

Now let us start the proof of {\it (i)}. Constants only depending
on $\epsilon$ will be named $C_\epsilon$ even though they might have
different values. 

Remember that it is enough to prove the stability result near
$b_0=0$. Let $K\in \mathcal{H}^\alpha$ with $\alpha>3$. Lemma
\ref{lema1} gives us some $\beta>2$ such that $K\in
\mathcal{C}^\beta$. Proposition
\ref{prop1} says that $(K-z)^{-1}\in \mathcal{H}^{\beta'}$ with some
$2<\beta'<\beta $, while Lemma
\ref{lema1} insures that there exists $\gamma >1$ such that 
$(K-z)^{-1}\in \mathcal{C}^\gamma$.

Denote by $G(\x,\x';z)$ the kernel
of $(K-z)^{-1}$. From \eqref{april8} and \eqref{april1}
we obtain a constant $C_\epsilon$ such that: 
\begin{equation}\label{sept3}
\sup_{\x'\in\Gamma}\sum_{\x\in \Gamma}\langle \x-\x'\rangle
|G(\x,\x';z)|\leq C_\epsilon\quad {\rm if}\quad {\rm
  dist}(z,\sigma(K))\geq \epsilon.
\end{equation}
Define the operator $S_{b}(z)$ to be the one corresponding to
the kernel $e^{ib\varphi(\x,\x')}G(\x,\x';z)$. Using
the Schur-Holmgren criterion we can write 
$$||S_{b}(z)||\leq C_\epsilon,\quad b\in\R,\quad {\rm
  dist}(z,\sigma(K))\geq \epsilon .$$

Using \eqref{iunie2}
we can write:
\begin{equation}\label{sept4}
(K_{b}-z)S_{b}(z)=:1+T_{b}(z),
\end{equation}
where $T_{b}(z)$ is given by the kernel
\begin{equation}\label{sept5}
e^{ib\varphi(\x,\x')}\sum_{\y\in\Gamma}(e^{ib\varphi(\x-\y,\x'-\y)}-1)K(\x,\y)\; 
G(\y,\x';z).
\end{equation}
Note that 
\begin{equation}\label{sept6}
|e^{ib\varphi(\x-\y,\x'-\y)}-1|\leq |b|\;|\varphi(\x-\y,\x'-\y)|\leq
\frac{|b|}{2}\; |\x-\y|\; |\y-\x'|.
\end{equation}
Then for any $f\in l^2(\Gamma)$ with compact support we can write:
 \begin{equation}\label{sept7}
|T_{b}(z)f|(\x)\leq |b|
\sum_{\y\in\Gamma}|\x-\y|\;|K(\x,\y)|\; 
|\y-\x'|\;|G(\y,\x';z)|\; |f(\x')|
\end{equation}
and after applying the Schur-Holmgren criterion we get: 
$$||T_b(z)||\leq |b|\; ||K||_{\mathcal{C}^1}
||(K-z)^{-1}||_{\mathcal{C}^1}\leq |b|\; C_\epsilon.$$
Thus if $|b|$ is small enough, $||T_b(z)||\leq 1/2$ whenever ${\rm
  dist}(z,\sigma(K))\geq \epsilon$. Now if Im$z\neq 0$ we know that 
$K_{b}-z$ is invertible, and from 
\eqref{sept4} we conclude that there exists a constant $C_\epsilon$ such that
\begin{align}\label{sept8}
(K_{b}-z)^{-1}&=S_{b}(z)\;(1+T_{b}(z))^{-1},\nonumber \\
||(K_{b}-z)^{-1}||&\leq C_\epsilon \quad {\rm whenever}\;
|b|\leq b_\epsilon \; {\rm and}\; {\rm
  dist}(z,\sigma(K))\geq \epsilon, 
\end{align}
uniformly in the imaginary part of $z$. 
This means that ${\rm
  dist}(z,\sigma(K_b))\geq \frac{1}{C_\epsilon}>0$ whenever 
$|b|\leq b_\epsilon$ and ${\rm
  dist}(z,\sigma(K))\geq \epsilon$, and the proof of {\it (i)} is over. \qed

\subsection{Proof of {\it (ii)}}
As before, we only need to consider $b_0=0$. We give the proof just
for the upper spectral limit $E_+$, since the argument for $E_-$ is
similar. 

\subsubsection{Reduction to localized operators}
We start with an abstract lemma. 
\begin{lemma}\label{lemasupreme}
Let $M(b)$ and $N(b)$ be two families of bounded and self-adjoint operators on some
Hilbert space $\mathcal{H}$, such that $||M(b)-N(b)||\leq C\; |b|$ if
$|b|\leq 1$. Then: 
\begin{align}\label{noapte1}
|\sup\sigma(M(b))-\sup\sigma(N(b))|
\leq ||M(b)-N(b)||\leq C\; |b|,\quad |b|\leq 1,
\end{align}
and a similar estimate holds for the infimum of their spectra. In
particular, if $\sup\sigma(N(b))$ is Lipschitz at $b=0$ then the same
is true for $\sup\sigma(M(b))$. 
\end{lemma}
\noindent{\bf Remark}. Note the important thing that we do not require
from $M(b)$ and $N(b)$ to converge in norm to $M(0)=N(0)$ when
$b$ tends to zero. 
\begin{proof}
For every $\psi \in \mathcal{H}$ with $||\psi||=1$ we can write
$$\langle M(b)\psi,\psi\rangle\leq \langle N(b)\psi,\psi\rangle +
||M(b)-N(b)||\leq
\sup\sigma(N(b))+||M(b)-N(B)||$$
which means that $\sup\sigma(M(b))-\sup\sigma(N(b))\leq ||M(b)-N(b)||$. By
interchanging $M(b)$ with $N(b)$ we obtain the inequality:
\begin{align}\label{sept9}
|\sup\sigma(M(b))-\sup\sigma(N(b))|\leq ||M(b)-N(b)||.
\end{align}
A similar argument shows the same estimate for the infimum of the
spectra. Regarding the Lipschitz property, we use that
$\sup\sigma(M(0))=\sup\sigma(N(0))$ and then we apply the triangle
inequality:
\begin{align}\label{septembrie1}
|\sup\sigma(M(b))-\sup\sigma(M(0))|\leq
|\sup\sigma(N(b))-\sup\sigma(N(0))|+||M(b)-N(b)||\leq C\; |b|.
\end{align}

\end{proof}

\vspace{0.5cm}

Getting back to our theorem, we now want to reduce the problem to operators with
kernels supported near the diagonal. Denote by $\chi$ the characteristic function of the interval
$[0,1]$. Denote by $\widehat{K}_b$ the operator given by the 
kernel $\widehat{K}_b(\x,\x'):=
\chi(\sqrt{b}|\x-\x'| ) K(\x,\x')$ and
by $\widetilde{K}_b$ the operator given by 
$\widetilde{K}_b(\x,\x'):=\chi(\sqrt{b}|\x-\x'| )
e^{ib\varphi(\x,\x')}K(\x,\x')$.

Since $K\in \mathcal{H}^\alpha$ with $\alpha>3$, according to Lemma
\ref{lema1} we have the bound:
\begin{align}\label{sept10}
\sup_{\x'\in\Gamma}\sum_{\x\in\Gamma}\langle \x-\x'\rangle^2
|K(\x,\x')|=||K||_{\mathcal{C}^2}<\infty. 
\end{align}

Via the Schur-Holmgren criterion we obtain: 
\begin{align}\label{sept11}
\max\{||K-\widehat{K}_b||,\;||K_b-\widetilde{K}_b||\}\leq 
\sup_{\x'\in\Gamma}\sum_{\x\in\Gamma}
\left [1-\chi(\sqrt{b}|\x-\x'| )\right ]
|K(\x,\x')|\leq |b|\; ||K||_{\mathcal{C}^2}. 
\end{align}
Using Lemma \ref{lemasupreme} for the pair $K$ and $\widehat{K}_b$ we
obtain $|E_+(0)-\sup(\sigma(\widehat{K}_b))|\leq |b|\;
||K||_{\mathcal{C}^2}$. The same lemma for the pair $K_b$ and
$\widetilde{K}_b$ gives $|E_+(b)-\sup(\sigma(\widetilde{K}_b))|\leq |b|\;
||K||_{\mathcal{C}^2}$. Then the triangle inequality leads to: 
\begin{align}\label{sept12}
|E_+(b)-E_+(0)|\leq 2|b|\; ||K||_{\mathcal{C}^2}+
|\sup(\sigma(\widetilde{K}_b))-\sup(\sigma(\widehat{K}_b))|. 
\end{align}
Thus we have reduced the problem to the study of the spectral edges of 
$\widetilde{K}_b$ and $\widehat{K}_b$.

\subsubsection{Study of the operators with cut-off}

Clearly, $\widetilde{K}_b(\x,\x')=
e^{ib\varphi(\x,\x')}\widehat{K}_b(\x,\x')$. Without loss, assume that
$b>0$. 
Take $\psi\in l^2(\Gamma)$ with compact support and compute 
(use \eqref{ianu3} in the second equality):
\begin{align}\label{ianu5}
&\langle \widetilde{K}_b\psi,\psi\rangle=\sum_{\x,\x'\in
  \Gamma}e^{ib\varphi(\x,\x')}\widehat{K}_b(\x,\x')\psi(\x')\overline{\psi(\x)}\nonumber
\\
&=
\int_{{\bf R}^2}d{\bf y}\sum_{\x,\x'\in
  \Gamma}\psi(\x')\overline{\psi(\x)}\frac{4\pi\sinh (2 bt)}{b}
\widehat{K}_b(\x,\x')\exp{\left [\frac{b|{\bf x}-{\bf x}'|^2}{4\tanh ( 2bt)}\right ]}
G_b(\x,\y;t)G_b(\y,\x';t).
\end{align}
Now denote by $A_b(t)$ the operator with kernel $$A_b(\x,\x';t):=
\widehat{K}_b(\x,\x')\exp{\left [\frac{b|{\bf x}-{\bf x}'|^2}{4\tanh (
      2bt)}\right ]}=K(\x,\x')\exp{\left [\frac{b|{\bf x}-{\bf
        x}'|^2}{4\tanh ( 2bt)}\right ]}\chi (\sqrt{b}|\x-\x'|).$$

The crucial observation is that equation \eqref{ianu5} leads to:
\begin{align}\label{ianu8}
&\langle \widetilde{K}_b\psi,\psi\rangle=
\int_{{\bf R}^2}d{\bf y}\langle A_b(t)  G_b(\y,\cdot;t)\psi, G_b(\y,\cdot;t)\psi\rangle 
\frac{4\pi\sinh (2 bt)}{b}\nonumber \\
&\leq \sup\sigma(A_b(t))\frac{4\pi\sinh (2 bt)}{b}
\int_{{\bf R}^2}d{\bf y}|| G_b(\y,\cdot;t)\psi||^2\nonumber \\
&=
\sup\sigma(A_b(t))\frac{4\pi\sinh (2 bt)}{b}\int_{{\bf R}^2}d{\bf
  y}\sum_{\x\in\Gamma} |G_b(\y,\x;t)|^2|\psi(\x)|^2\nonumber \\
&=\sup\sigma(A_b(t)) \; ||\psi||^2,
\end{align}
where in the last line we used \eqref{ianu4}. It means that
$\sup\sigma(\widetilde{K}_b)\leq \sup(\sigma(A_b(t)))$ for all $t$. Now
let us show that the operator
$A_b(t)-\widehat{K}_b$ has a norm proportional with $b$ if $t$ is large
enough (say $t=b^{-1}$). Indeed, we can write
\begin{align}\label{ianu9}
&|A_b(\x,\x';b^{-1})-\widehat{K}_b(\x,\x')|\leq |K(\x,\x')|\chi(
  \sqrt{b}  |\x-\x'|)
\left (\exp{\left [\frac{b|{\bf x}-{\bf x}'|^2}{4\tanh (2)}\right
    ]}-1\right )\nonumber \\
&\leq |K(\x,\x')|\chi(\sqrt{b}
    |\x-\x'|)\frac{b|{\bf x}-{\bf x}'|^2}{4\tanh (2)}
\exp{\left [\frac{b|{\bf x}-{\bf x}'|^2}{4\tanh ( 2)}\right
  ]}
\end{align}
and on the support of $\chi$ we can bound the above difference
with:
\begin{align}\label{ianu10}
&|A_b(\x,\x';b^{-1})-\widehat{K}_b(\x,\x')|\leq {\rm const}\; b\; 
|\x-\x'|^2|K(\x,\x')|.
\end{align}
The right hand side defines an operator whose norm behaves like $b$. 
Thus \eqref{ianu8} and \eqref{ianu10} imply:
\begin{align}\label{ianu11}
\sup\sigma(\widetilde{K}_b)\leq \sup\sigma(A_b(b^{-1}))\quad {\rm and}\quad 
||A_b(b^{-1})-\widehat{K}_b||\leq C\;b.
\end{align}
Using \eqref{sept9} for the pair $A_b(b^{-1})$ and $\widehat{K}_b$ we arrive at:
\begin{align}\label{ianu12}
\sup\sigma(\widetilde{K}_b)\leq
\sup\sigma(\widehat{K}_b)+ C\;b.
\end{align}

We now want to change places between $\widetilde{K}_b$ and $\widehat{K}_b$ in the above
inequality, which would lead to $\sup\sigma(\widehat{K}_b)\leq
\sup\sigma(\widetilde{K}_b)+C\;b$ and thus:
$$|\sup\sigma(\widetilde{K}_b)-
\sup\sigma(\widehat{K}_b)|\leq C\;b,$$
which together with \eqref{sept12} would imply:
$$|E_+(b)-E_+(0)|\leq C\; b,\quad b\geq 0.$$

The key step in the proof of \eqref{ianu12} was \eqref{ianu5}. Since 
$\widehat{K}_b(\x,\x')=
e^{-ib\varphi(\x,\x')} \widetilde{K}_b(\x,\x')$ we can
write (use \eqref{ianu3'} in the second line):
\begin{align}\label{ianu5'}
&\langle \widehat{K}_b\psi,\psi\rangle=\sum_{\x,\x'\in
  \Gamma}e^{-ib\varphi(\x,\x')}\widetilde{K}_b(\x,\x')\psi(\x')
\overline{\psi(\x)}\nonumber
\\
&=
\int_{{\bf R}^2}d{\bf y}\sum_{\x,\x'\in
  \Gamma}\psi(\x')\overline{\psi(\x)}\frac{4\pi\sinh (2 bt)}{b}
\widetilde{K}_b(\x,\x')\exp{\left [\frac{b|{\bf x}-{\bf x}'|^2}{4\tanh ( 2bt)}\right ]}
G_b(\x',\y;t)G_b(\y,\x;t).
\end{align}
Now everything will work as before, because the phase
$e^{ib\varphi(\x,\x')}$ changes neither the localization nor the
$\mathcal{C}^2$ norm of the operators. The proof for the upper
spectral edges is over. 

The proof for the lower spectral edges is based on an estimate which
is very similar with \eqref{ianu8}, in which we reverse the inequality
and show that $\inf\sigma(\widetilde{K}_b)\geq \inf\sigma(A_b(t))$ for
all $t$. We give no further details.

\subsection{Proof of {\it (iii)}}

The idea is to reduce the problem to the previous case. Again it is
enough to consider $b_0=0$ and $b>0$ small enough. Assume that $K$ has
a gap in its spectrum of the form $(e_-,e_+)$, with $e_\pm\in
\sigma(K)$. Then due to {\it (i)} we know that if $b$ is small enough
the gap will survive: we can choose a positively oriented 
circle $L$ in the complex plane
containing  $\Sigma_+(b):=\sigma(K_b)\cap [e_+(b),\infty)$ such that 
$${\rm dist}(z,\sigma(K_b))\geq \eta>0 \quad {\rm whenever} \quad z\in L \quad
{\rm and}\quad 0<b<b_\eta. $$

The orthogonal projector $P_b$ corresponding to $\Sigma_+(b)$ can be written
as a Riesz integral and we have:

\begin{align}\label{septe1}
P_b:=\frac{i}{2\pi}\int_L(K_b-z)^{-1}dz,\quad
K_bP_b=\frac{i}{2\pi}\int_Lz(K_b-z)^{-1}dz,\quad b\geq 0.
\end{align}
If we consider $K_bP_b$ as an operator living on the whole space
$l^2(\Gamma)$, then its spectrum is given by the union $\{0\}\cup
\Sigma_+(b)$. If we choose $\lambda:=1+\sup\sigma(K)$, then for $b$ small
enough the operator 
$ D_b:=K_bP_b-\lambda P_b$ will have $\inf \sigma (D_b)=
e_+(b)-\lambda\leq -1/2$. Thus $e_+(b)=\lambda+\inf\sigma(D_b)$, hence
$e_+(b)$ is Lipschitz at $b=0$ if $\inf\sigma(D_b)$ has the same
property. This is what we prove next:
\begin{lemma}\label{lemma13}
Let $D_b=K_bP_b-\lambda P_b$ with $\lambda:=1+\sup\sigma(K)$. Then
there exists $b_1>0$ small enough and a constant $C>0$ such that for
every 
$0< b<b_1$ we have $ |\inf\sigma(D_b)-\inf\sigma(D_0)|\leq C\; b$. 
\end{lemma}
\begin{proof}
Remember that we imposed $\alpha>4$. We have that 
$||K_b||_{\mathcal{H}^\alpha}=||K||_{\mathcal{H}^\alpha}<\infty$ for
all $b$. According to Lemma \ref{lema1}, there exists $\beta>3$ such
that
$||K_b||_{\mathcal{C}^{\beta}}=||K||_{\mathcal{C}^{\beta}}<\infty$. Then
if $b$ is smaller than some constant only depending on $L$, Proposition
\ref{prop1} tells us that $(K_b-z)^{-1}\in {\mathcal{H}^{\beta'}}$ for some
$3<\beta'<\beta$, for all
$z\in L$ and 
$\sup_{z\in L}||(K_b-z)^{-1}||_{\mathcal{H}^{\beta'}}\leq C.$ Thus
both $P_b$ and $D_b$ belong to $\mathcal{H}^{\beta'}$ with $\beta'>3$ 
if $b$ is small
  enough. More
  precisely, there exists $b_2>0$ sufficiently small such that 
\begin{align}\label{septe2}
\max\{||P_b||_{\mathcal{H}^{\beta'}},||D_b||_{\mathcal{H}^{\beta'}}\}\leq C,\quad
0\leq b\leq b_2.
\end{align}
If $G(\x,\x';z)$ is the integral kernel of $(K-z)^{-1}$, then we
introduced at point {\it (i)} the operator $S_b(z)$ given by the
kernel $e^{ib\varphi(\x,\x')}G(\x,\x';z)$. Using \eqref{sept8} 
we can write: 
\begin{align}\label{septe3}
\sup_{z\in L}||(K_b-z)^{-1}-S_b(z)||\leq C\;b,
\end{align}
provided $b$ is small enough. Denoting by $D_0$ the operator given by
the integral kernel
$$D_0(\x,\x'):=\frac{i}{2\pi}\int_L(z-\lambda)G(\x,\x';z)dz$$
and by $(D_0)_b$ the operator generated by $e^{ib\varphi(\x,\x')}D_0(\x,\x')$,
then using \eqref{septe3} we arrive at the estimate:
\begin{align}\label{septe4}
||D_b-(D_0)_b||\leq C\;b\quad {\rm whenever}\quad 0\leq b<b_2.
\end{align}
It follows from Lemma \ref{lemasupreme} that $\inf\sigma(D_b)$ is Lipschitz at $b=0$ if 
$\inf\sigma((D_0)_b)$ has the same property. But for the operator
$(D_0)_b$ we can apply point {\it (ii)}, and the proof is over.
\end{proof}

\section{Proof of Corollary \ref{corolar1}}

In order to keep the notation simple, we will only consider
$b_0=0$. Here the generating kernel $K(\x,\x';b)$ depends on $b$ and
\eqref{aprilie1} at $b_0=0$ reads as:
\begin{equation}\label{septembrie2}
\sup_{\x'\in\Gamma}\sum_{\x\in \Gamma}|K(\x,\x';b)-K(\x,\x';0)|\leq
C\; |b|,\quad |b|\leq 1.
\end{equation}
Let us introduce the family $\widetilde{K}_b$ where their kernels are
given by $e^{ib\varphi(\x,\x')}K(\x,\x';0)$. Clearly,
$K_0=\widetilde{K}_0$. Moreover, \eqref{septembrie2} implies that
 $||K_b-\widetilde{K}_b||\leq C\;|b|$ around $b=0$. We know that Theorem
 \ref{teorema1} {\it (ii)} applies for
$\widetilde{K}_b$ around $b=0$, so the only thing we have left is to
extend it to $K_b$. From Lemma \ref{lemasupreme} we immediately conclude that
$\sup \sigma(K_b)$ and $\inf \sigma(K_b)$ are Lipschitz at $b=0$. 

The spectral stability of $K_b$ can be shown with the same strategy as
the one one used in \eqref{sept3}-\eqref{sept8}. The operator
$S_b(z)$ must be constructed starting from the kernel of
$(K_0-z)^{-1}$ which gets multiplied with the phase
$e^{ib\varphi(\x,\x')}$. When we act with $K_b-z$ on $S_b(z)$ as in
\eqref{sept4}, we obtain an extra term which enters in $T_b(z)$, which
is $(K_b-\widetilde{K}_b)S_b(z)$. This error is again proportional with
$|b|$ if $z$ is at some distance from the spectrum of $K_0$. Thus
\eqref{sept8} holds again. 

For the case of gaps, the proof is identical with the case independent
of $b$.  
  \qed

\section{Proof of Theorem \ref{teorema2}}

There are important similarities between the proof strategies in the discrete and continuous
cases. Although the stability of the resolvent set of $H(b)$ is known,
we will sketch a short proof which will also provide some ingredients
for the proof of the Lipschitz behavior of the band edges.

\subsection{Stability of gaps}
Assume that
$M\subset \rho(H(0))$ is a compact set and ${\rm dist}(M,\sigma(H(0))) >0$. The resolvent $(H(0)-z)^{-1}$ is an integral
operator given by an integral kernel $Q_0(\x,\x';z)$. 

The singularities of $Q_0(\x,\x';z)$ are the same as in
the case of the free Laplacean and there exists
some $\delta>0$ and $C_M<\infty$ such that uniformly in $\x\neq \x'$
\cite{GK, CN3}: 
  \begin{align}\label{decem2}
&\sup_{z\in M}|Q_0(\x,\x';z)|\leq C_M\;
(1+|\ln(|\x-\x'|)|)e^{-\delta|\x-\x'|},\nonumber \\ 
&\sup_{z\in M}|\nabla_\x Q_0(\x,\x';z)|\leq C_M\left (1+
  \frac{1}{|\x-\x'|}\right )e^{-\delta|\x-\x'|}. 
\end{align}
In particular we have the following Schur-Holmgren type property:
\begin{equation}\label{decem1}
\sup_{z\in M}\sup_{\x'\in\R^2}\int_{\R^2}|Q_0(\x,\x';z)|d\x \leq C(M)<\infty.
\end{equation}
This allows us to define for every $z\in M$ a bounded operator
$S_b(z)$ whose integral kernel is given by:
\begin{equation}\label{oprima40}
S_b(\x,\x';z):=e^{ib\varphi(\x,\x')}Q_0(\x,\x';z),\quad \sup_{z\in
  M}||S_b(z)||\leq C(M)<\infty.
\end{equation}

Define $T_b(z)$ to be the operator with the integral kernel:
\begin{equation}\label{oprima41}
T_b(\x,\x';z):=be^{ib\varphi(\x,\x')}\{2i {\bf a}(\x-\x')\nabla_\x Q_0(\x,\x';z)+b|{\bf a}(\x-\x')|^2         Q_0(\x,\x';z)\}.
\end{equation}
The kernel $T_b(\x,\x';z)$ is bounded because the inequality $|{\bf a}(\x-\x')|\leq |\x-\x'|$ 
compensates 
the local singularities of $\nabla_\x Q_0(\x,\x';z)$ and
$Q_0(\x,\x';z)$ when $|\x-\x'|$ is small,
while when $|\x-\x'|$ is large we have the exponential decay which
comes into play. In fact, using
\eqref{decem2} we see that the kernel $T_b(\x,\x';z)$ obeys a
Schur-Holmgren estimate. We get:
\begin{equation}\label{oprima7}
 \sup_{z\in
  M}||T_b(z)||\leq C(M) \;|b|,\quad |b|\leq 1.
\end{equation}

Note the important identity valid on Schwartz functions:
\begin{equation}\label{oprima4}
\{-i\nabla_{\x}-b{\bf
  a}(\x)\}e^{ib\varphi(\x,\x')}=e^{ib\phi(\x,\x')}\{-i\nabla_{\x}-b{\bf
  a}(\x-\x')\}.
\end{equation}
Let us note that $S_b(z)$ leaves the Schwartz space invariant and for
such two functions $f$ and $g$  
we have (using \eqref{oprima4}):
\begin{align}\label{oprima5}
&\langle \{({\bf p}-b{\bf a})^2+V-z\}S_b(z)f,g\rangle \\
&=\int_{\R^2}\int_{\R^2}e^{ib\varphi(\x,\x')}(\{-i\nabla_\x-b{\bf a}(\x-\x')\}^2+V(\x)-z)
Q_0(\x,\x';z)f(\x')\overline {g(\x)}d\x d\x'\nonumber \\
&=\langle f,g\rangle +\langle T_b(z)f,g\rangle.\nonumber 
\end{align}
The operator $H(b)$ is essentially self-adjoint on the Schwartz space,
and after a density argument we conclude that the range of $S_b(z)$ is
contained in the domain of $H(b)$ and $(H(b)-z)S_b(z)=1+T_b(z)$. Now
there exists $b_1>0$ small enough such that if $|b|\leq b_1$ we have $\sup_{z\in
  M}||T_b(z)||\leq 1/2$ (see \eqref{oprima7}). Then after a standard argument we conclude
\begin{equation}\label{oprima6}
(H(b)-z)^{-1}=S_b(z)(1+T_b(z))^{-1}, \quad \sup_{z\in
  M}||(H(b)-z)^{-1}||\leq C_M,\quad |b|\leq b_1.
\end{equation}
This means that the gaps in the spectrum of $H(0)$ are
preserved. In particular,
for every $\epsilon>0$ there exists $b_2(\epsilon)>0$ such that:
\begin{equation}\label{oprima60}
s_-(0)-\epsilon \leq s_-(b)\leq s_+(b)\leq s_+(0)+\epsilon\quad {\rm
  whenever }\quad |b|\leq b_2(\epsilon).
\end{equation}
Choose a positively oriented circle
$L$ isolated from $\sigma(H(0))$ such that $L$ completely contains 
the finite band
$\sigma_0$. Then if $|b|$ is small enough $L$ will completely contain
$\sigma_b$ and remain separated from $\sigma(H(b))$. 

\subsection{The reduction to Harper-like operators}
As in the discrete case, we construct the Riesz integrals
$$P_b:=\frac{i}{2\pi}\int_L (H(b)-z)^{-1}dz,\quad K(b):=H(b)P_b=
\frac{i}{2\pi}\int_L z(H(b)-z)^{-1}dz.$$
The operator $H(b)P_b$ seen in the whole space $L^2(\R^2)$ will have the
spectrum $\sigma_b\cup \{0\}$. Fix $\lambda_+:=1-s_-(0)$. If
$|b|\leq b_2(\frac{1}{2})$ (see \eqref{oprima60}), then we know that 
$s_+(b)+\lambda_+\geq s_-(b)+\lambda_+ \geq
\frac{1}{2}>0$. It means that 
$$s_+(b)+\lambda_+=\sup\sigma\{H(b)P_b+\lambda_+ P_b\}.$$
 Similarly,
choosing $\lambda_-:=-1-s_+(0)$ we have $s_-(b)+\lambda_-<-\frac{1}{2}<0$
hence 
$$s_-(b)+\lambda_-=\inf\sigma\{H(b)P_b+\lambda_- P_b\}.$$  
In other words, the band edges $s_\pm$ will be Lipschitz at $b=0$ if 
the spectral
edges of the
operators
$$K_\pm(b):= H(b)P_b+\lambda_\pm P_b=\frac{i}{2\pi}\int_L (z+\lambda_\pm)(H(b)-z)^{-1}dz$$
have the same property. Note that the operator $K_\pm(0)$ has an
integral kernel given by:
\begin{align}\label{decem10}
K_\pm(0)(\x,\x')=\frac{i}{2\pi}\int_L
(z+\lambda_\pm)Q_0(\x,\x';z)dz,\quad |K_\pm(0)(\x,\x')|\leq
Ce^{-\delta |\x-\x'|},
\end{align}
where the local singularity at $\x=\x'$ dissapears due to the
integral with respect to $z$.

Now using \eqref{oprima6}, \eqref{oprima40}
and \eqref{oprima7} we have:
\begin{align}\label{decem9}
\left \Vert K_\pm(b)-\frac{i}{2\pi}\int_L
  (z+\lambda_\pm)S_b(z)dz\right \Vert \leq C\; |b|.
\end{align}
According to Lemma \ref{lemasupreme}, the spectral edges of 
$K_\pm(b)$ are Lipschitz at $b=0$ if the same
property is true for 
$$\left (K_\pm(0)\right )_b:=\frac{i}{2\pi}\int_L
  (z+\lambda_\pm)S_b(z)dz.$$
This notation wants to highlight the fact that $\left (K_\pm(0)\right )_b$
is given by the integral kernel
$$\left (K_\pm(0)\right
)_b(\x,\x'):=e^{ib\varphi(\x,\x')}K_\pm(0)(\x,\x').$$
At this point we are in a situation which is completely similar to the
discrete case, with the difference that the Hilbert space is
$L^2(\R^2)$ and the sums over $\Gamma$ have to be replaced by
integrals. The unperturbed kernel $K_\pm(0)(\x,\x')$ has an
exponential localization. 

We can mimic the proof of Theorem \ref{teorema1} {\it (ii)} and
conclude that the spectral edges of  $\left (K_\pm(0)\right )_b$ are
Lipschitz at $b=0$, and we are done. \qed

\section{Appendix}

\subsection{ Proof of Proposition \ref{prop1}}

 Denote by
$G(\x,\x';z)$ the integral kernel of $(H-z)^{-1}$. If $\alpha'=0$ we have 

$$\sum_{\x\in
    \Gamma}|G(\x,\x';z)|^2=||(H-z)^{-1}\delta_{\x'}||^2\leq \frac{1}
{\{{\rm dist}(z,\sigma(H))\}^{2}}$$ uniformly in $\x'$, an estimate
which is in fact much better than \eqref{april8}. So from now on we
may assume that $0<\alpha'<\alpha$.  

For $\bk\in \R^2$ define the unitary multiplication operator $U_\bk$
by $(U_\bk f)(\x)=e^{i\bk\cdot \x}f(\x)$. Define the family of 
isospectral operators $H_\bk=U_\bk HU_\bk^* $, with integral
kernels given by $H_\bk(\x,\x')=e^{i\bk\cdot
  (\x-\x')}H(\x,\x')$. We need the following technical result:
\begin{lemma}\label{lema10}
Let $H$ be an element of $ \mathcal{C}^\alpha$. Let $n$ be the integer
part of $\alpha$. Then the mapping  
$$\R^2\ni \bk\mapsto H_\bk\in B(l^2(\Gamma))$$
is $n$ times continuously differentiable in the norm topology. 
Moreover, any $n$'th
order mixed partial derivative of $H_\bk$ is $\alpha-n$ H\"older
continuous at $\bk=0$ in the norm topology. 
\end{lemma}
\begin{proof} Assume that $\bk=(k_1,k_2)$. 
The integral kernel of $H_\bk$ is $e^{i\bk\cdot
    (\x-\x')}H(\x,\x')$. Let $n$ be the integer part of $\alpha$. Then
  $H_\bk$ is $n$ times differentiable in the norm topology with respect to 
$k_j$, $j\in\{1,2\}$, and its
  $n$'th mixed partial derivative $\partial_{k_1}^m\partial_{k_2}^{n-m}H_\bk$ is given by the integral kernel $i^n(x_1-x'_1)^m(x_2-x'_2)^{n-m}e^{i\bk\cdot
    (\x-\x')}H(\x,\x')$. This integral kernel defines a bounded 
  operator because $|(x_1-x'_1)^m(x_2-x'_2)^{n-m}|\leq \langle
  \x-\x'\rangle ^n$ and then we can use \eqref{febr1}. 

For the H\"older continuity statement, we use the estimate 
$|e^{i\bk\cdot
    (\x-\x')}-1|\leq 2^{1-\beta}\; |\bk|^\beta |\x-\x'|^{\beta}$ which holds for
  every $0\leq \beta\leq 1$. 

\end{proof}

Now let $z\in\rho(H)$. Denote by
$G_\bk(\x,\x';z)$ the integral kernel of $(H_\bk-z)^{-1}$. Due to the
identity $U_\bk(H-z)^{-1}U_\bk^*=(H_\bk -z)^{-1}$ we have:
\begin{equation}\label{april2}
G_\bk(\x,\x';z)=e^{i\bk\cdot (\x-\x')}G(\x,\x';z).
\end{equation} 
Let us denote by $n$ the integer part of $\alpha$. We can suppose that
$n\geq 1$ since the case $0<\alpha<1$ is covered by the argument
below. 

From the identity
\begin{equation}\label{april5}
(H_{\bk'} -z)^{-1}-(H_{\bk}-z)^{-1}=-(H_{\bk'}
-z)^{-1}[H_{\bk'}-H_\bk](H_{\bk}-z)^{-1}
\end{equation} 
and from Lemma \ref{lema10} we conclude that the map 
$$\R^2\ni \bk\mapsto (H_\bk-z)^{-1}\in B(l^2(\Gamma))$$
is continuous in the norm topology, and also differentiable. We have:
\begin{equation}\label{april3}
D_\bk (H_{\bk} -z)^{-1}=-(H_\bk-z)^{-1}[D_\bk H_\bk](H_\bk-z)^{-1}.
\end{equation} 
Using this identity at $\bk=0$ in \eqref{april2} leads to:
$$(\x-\x')G(\x,\x';z)=-\langle (H-z)^{-1}[D_\bk
H_\bk]_{\bk=0}(H-z)^{-1}\delta_{\x'},\delta_\x\rangle $$
which gives:
$$||(H-z)^{-1}||_{\mathcal{H}^1}\leq C\;(1+||H||_{\mathcal{C}^1})\left (\frac{1}{{\rm
    dist}(z,\sigma(H))^2} +\frac{1}{{\rm
    dist}(z,\sigma(H))}\right ).$$
This is true because we have the pointwise bound 
\begin{align}\label{iulie1}
\langle \x-\x'\rangle^{2\alpha} &\leq  (1+|x_1-x_1'|+|x_2-x_2'|)^{2\alpha}\leq
(3\max\{1,|x_1-x_1'|,|x_2-x_2'|)^{2\alpha}\nonumber \\
&\leq 3^{2\alpha}+ \sum_{j=1}^2 3^{2\alpha}|x_j-x_j'|^{2\alpha}.
\end{align}

By induction we obtain the following rough estimate:
\begin{equation}\label{april4}
||(H-z)^{-1}||_{\mathcal{H}^n}\leq C_n\;(1+||H||_{\mathcal{C}_n}^n)\left (\frac{1}{{\rm dist}(z,\sigma(H))^{n+1}}
  +\frac{1}{{\rm
    dist}(z,\sigma(H))}\right ).
\end{equation}

Now let us assume that $n<\alpha<n+1$. The integral kernel of the
$n$'th partial derivative of $(H_{\bk} -z)^{-1}$ with respect to $k_1$
is given by $i^ne^{i\bk\cdot
  (\x-\x')}(x_1-x_1')^nG(\x,\x';z)$. Moreover, using \eqref{april3}
and Lemma \ref{lema10} we
conclude that the operator $\partial_{k_1}^n(H_{\bk} -z)^{-1}$ is
$\alpha-n$ H\"older continuous at $\bk=0$. Let $\bk=(k_1,0)$. 
We also have the identity:
$$i^n(e^{ik_1(x_1-x_1')}-1)(x_1-x_1')^nG(\x,\x';z)=\langle 
[\partial_{k_1}^n(H_{k_1} -z)^{-1}-\partial_{k_1}^n(H_{k_1}
-z)^{-1}|_{\bk=0}]\delta_{\x'}, \delta_\x\rangle .$$ 
If $|k_1|\leq 1$ the following norm estimate holds true according to 
Lemma \ref{lema10}: 
 \begin{align}\label{april15}
&||[\partial_{k_1}^n(H_{k_1} -z)^{-1}-\partial_{k_1}^n(H_{k_1}
-z)^{-1}|_{\bk=0}]||\nonumber \\
&\leq C |k_1|^{\alpha-n}(1+||H||_{\mathcal{C}^\alpha}^{n+1})\left (\frac{1}
{{\rm dist}(z,\sigma(H))^{n+2}}
+\frac{1}{{\rm
    dist}(z,\sigma(H))}\right).
\end{align}
Choose $n<\alpha'<\alpha<n+1$. Then the following integral converges
in norm and defines a bounded operator: 
$$\widetilde{H}:=\int_0^\infty \frac{1}{k_1^{1+\alpha'-n}}
[\partial_{k_1}^n(H_{k_1} -z)^{-1}-\partial_{k_1}^n(H_{k_1}
-z)^{-1}|_{\bk=0}]dk_1.$$
Its integral kernel is given by 
$$\widetilde{G}(\x,\x';z):=i^n(x_1-x_1')^nG(\x,\x';z)\int_0^\infty
\frac{1}
{k_1^{1+\alpha'-n}}(e^{ik_1(x_1-x_1')}-1)dk_1.$$
Assuming without loss of generality that $x_1-x_1'\neq 0$, and by a change
of variable $s=k_1\;|x_1-x_1'|$ we obtain:
$$\widetilde{G}(\x,\x';z)=|x_1-x_1'|^{\alpha'-n}(x_1-x_1')^n  G(\x,\x';z)\int_0^\infty
\frac{1}{s^{1+\alpha'-n}}i^n(e^{is\;{\rm sign}(x_1-x_1')}-1)ds.$$
Notice that the above integral only has two possible values $C_\pm$
both different from
zero, depending on the sign of $x_1-x_1$. 
Since $\widetilde{G}(\x,\x';z)=\langle
\widetilde{H}\delta_{\x'},\delta_{\x}\rangle
=C_\pm(x_1,x_1')\;|x_1-x_1'|^{\alpha'-n}(x_1-x_1')^n G(\x,\x';z)$ with
$|C_\pm(x_1,x_1')|\geq C$ 
it follows that 
$$\sup_{\x'\in\Gamma}\sum_{\x\in\Gamma}|x_1-x_1'|^{2\alpha'}
\;|G(\x,\x';z)|^2\leq C^{-2}
||\widetilde{H}||^2.$$
This argument can be repeated for the other coordinate and bound
the $l^2$ norm of $\langle \cdot -\x'\rangle^{\alpha'}
G(\cdot,\x';z)$ using \eqref{iulie1}. The proof of Proposition
\ref{prop1} is over.  

\vspace{0.5cm}

\subsection{A few identities from the continuous case}

We list here a few well known facts about the continuous two
dimensional magnetic Schr\"odinger operator with constant magnetic field equal to $b$
in $L^2(\R^2)$: 
\begin{align}\label{queryc}
H_b=({\bf p}-b{\bf a}({\bf x}))^2,\quad {\bf p}=-i\nabla_{\bf
  x},\quad{\bf a}({\bf x})=(-x_2/2,x_1/2 ).
\end{align}

The integral kernel of the semi-group $e^{-tH_b}$ is denoted with 
$G_b(\x,\x';t)$ and is given by the following explicit formula: 
\begin{align}\label{ianu1}
G_b(\x,\x';t)=e^{ib\varphi({\bf x},{\bf
    x}')}\frac{b}{4\pi\sinh ( bt)}
\exp{\left [-\frac{b|{\bf x}-{\bf x}'|^2}{4\tanh ( bt)}\right ]}
=: e^{ib\varphi({\bf x},{\bf
    x}')}\widetilde{G}_b(\x,\x';t).
\end{align}
The semigroup property insures the following identity:
\begin{align}\label{ianu2}
G_b(\x,\x';2t)=\int_{{\bf R}^2}G_b(\x,\y;t)G_b(\y,\x';t)d{\bf y}.
\end{align}
Then we can write:
\begin{align}\label{ianu3}
e^{ib\varphi({\bf x},{\bf
    x}')}&=\frac{1}{\widetilde{G}_b(\x,\x';2t)}
\int_{{\bf R}^2}G_b(\x,\y;t)G_b(\y,\x';t)d{\bf y}\nonumber \\
&=\frac{4\pi\sinh (2 bt)}{b}
\exp{\left [\frac{b|{\bf x}-{\bf x}'|^2}{4\tanh ( 2bt)}\right ]}
\int_{{\bf R}^2}G_b(\x,\y;t)G_b(\y,\x';t)d{\bf y}.
\end{align}
Taking the complex conjugation in both sides gives:
\begin{align}\label{ianu3'}
e^{-ib\varphi({\bf x},{\bf
    x}')}=\frac{4\pi\sinh (2 bt)}{b}
\exp{\left [\frac{b|{\bf x}-{\bf x}'|^2}{4\tanh ( 2bt)}\right ]}
\int_{{\bf R}^2}G_b(\y,\x;t)G_b(\x',\y;t)d{\bf y}.
\end{align}

Again the semi-group property gives that:
\begin{align}\label{ianu4}
\frac{b}{4\pi\sinh ( 2bt)}=G_b(\x,\x;2t)=
\int_{{\bf R}^2}G_b(\x,\y;t)G_b(\y,\x;t)d{\bf y}=\int_{{\bf R}^2}|G_b(\y,\x;t)|^2d{\bf y}
\end{align}
which is clearly $\x$ independent.

\vspace{0.5cm}

\noindent {\bf Acknowledgments.}  H.C. acknowledges support from the Danish 
F.N.U. grant {\it  Mathematical Physics}. The author is deeply
indebted to Gheorghe Nenciu for his encouragements and for 
many fruitful discussions.

\end{document}